\documentclass[letterpaper, 10 pt, conference]{ieeeconf}  

\IEEEoverridecommandlockouts                              

\usepackage{graphics} 
\usepackage{mathptmx} 
\usepackage{times} 
\usepackage{amsmath} 
\usepackage{amssymb}  
\usepackage{accents}
\usepackage{color}
\usepackage{cite}
\usepackage{hyperref}
\usepackage{cleveref}
\title{\LARGE \bf
More Information is Not Always Better:\\
Connections between Zero-Sum Local Nash Equilibria in\\ Feedback and Open-Loop Information Patterns
}

\author{Kushagra Gupta$^{1}$, Ross E. Allen$^{2}$, David Fridovich-Keil$^{3}$ and Ufuk Topcu$^{3}$
\thanks{DISTRIBUTION STATEMENT A. Approved for public release. Distribution is unlimited.
This material is based upon work supported by the Under Secretary of Defense for Research and Engineering under Air Force Contract No. FA8702-15-D-0001. Any opinions, findings, conclusions or recommendations expressed in this material are those of the author(s) and do not necessarily reflect the views of the Under Secretary of Defense for Research and Engineering.
© 2025 Massachusetts Institute of Technology.
Delivered to the U.S. Government with Unlimited Rights, as defined in DFARS Part 252.227-7013 or 7014 (Feb 2014). Notwithstanding any copyright notice, U.S. Government rights in this work are defined by DFARS 252.227-7013 or DFARS 252.227-7014 as detailed above. Use of this work other than as specifically authorized by the U.S. Government may violate any copyrights that exist in this work. This work was also supported by the following: Office of Naval Research (ONR) grants ONR N00014-22-1-2703 and ONR N00014-24-1-2797 as well as the National Science Foundation CAREER award under Grant No. 2336840.}
\thanks{$^{1}$Kushagra Gupta is with the Department of Electrical and Computer Engineering, The University of Austin at Texas, TX, 78712, USA
        {\tt\small kushagrag@utexas.edu}}%
\thanks{$^{2}$Ross E. Allen with Massachusetts Institute of Technology Lincoln Laboratory, Lexington, MA 02420, USA
        {\tt\small ross.allen@ll.mit.edu}}%
\thanks{$^{2}$David Fridovich-Keil and Ufuk Topcu are with the Department of Aerospace Engineering and the Oden Institute for Computational
Engineering and Sciences, The University of Austin at Texas, TX, 78712, USA
        {\tt\small dfk@utexas.edu, utopcu@utexas.edu}}%
}

\newtheorem{theorem}{Theorem}
\newtheorem{remark}{Remark}
\begin{document}

\maketitle
\thispagestyle{empty}
\pagestyle{empty}
\newcommand{\kushy}[1]{\textcolor{red}{kushy: #1}}
\newcommand{\david}[1]{\textcolor{blue}{david: #1}}
\newcommand{\ufuk}[1]{\textcolor{blue}{ufuk: #1}}
\newcommand{\doubt}[1]{\textcolor{red}{#1}}

\newcommand{\bb}[1]{\mathbb{#1}}

\newcommand{\x}{x}
\newcommand{\xt}{\x_t}
\newcommand{\xs}{\x_s}
\newcommand{\xtone}{\x_{t+1}}
\newcommand{\xsone}{\x_{s+1}}
\newcommand{\xttoend}{\x_{t:K+1}}
\newcommand{\xol}{\x_{OL}}
\newcommand{\xolt}{\x_{OL_{t}}}
\newcommand{\xoltime}[1]{\x_{OL_{#1}}}
\newcommand{\xkone}{\x_{K+1}}
\newcommand{\xstar}{\x^*}
\newcommand{\xstartime}[1]{\x^*_{#1}}

\newcommand{\control}{u}
\newcommand{\ui}{\control^i}
\newcommand{\umi}{\control^{-i}}
\newcommand{\uone}{\control^1}
\newcommand{\utwo}{\control^2}
\newcommand{\uit}{\ui_t}
\newcommand{\uis}{\ui_s}
\newcommand{\umis}{\umi_s}
\newcommand{\umit}{\umi_t}
\newcommand{\uonet}{\uone_t}
\newcommand{\utwot}{\utwo_t}
\newcommand{\uones}{\uone_s}
\newcommand{\utwos}{\utwo_s}
\newcommand{\uttoend}{\control_{t:K}}
\newcommand{\uiol}{\control^{i,\text{OL}}}
\newcommand{\uoneol}{\control^{1,\text{OL}}}
\newcommand{\utwool}{\control^{2,\text{OL}}}
\newcommand{\uiolt}{\uiol_t}
\newcommand{\uioltime}[1]{\uiol_{#1}}
\newcommand{\uifb}{\control^{i,\text{FB}}}
\newcommand{\uifbtstar}{\control^{i*,\text{FB}}}
\newcommand{\umifb}{\control^{-i,\text{FB}}}
\newcommand{\uonefb}{\control^{1,\text{FB}}}
\newcommand{\utwofb}{\control^{2,\text{FB}}}
\newcommand{\uifbt}{\uifb_t}
\newcommand{\umifbt}{\umifb_t}
\newcommand{\uifbtime}[1]{\uifb_{#1}}
\newcommand{\umifbtime}[1]{\umifb_{#1}}
\newcommand{\uonefbtime}[1]{\uonefb_{#1}}
\newcommand{\utwofbtime}[1]{\utwofb_{#1}}
\newcommand{\uonefbstar}{\control^{1*,\text{FB}}}
\newcommand{\utwofbstar}{\control^{2*,\text{FB}}}
\newcommand{\umioltime}[1]{\umiol_{#1}}
\newcommand{\uoneoltime}[1]{\uoneol_{#1}}
\newcommand{\utwooltime}[1]{\utwool_{#1}}
\newcommand{\uoneolstar}{\control^{1*,\text{OL}}}
\newcommand{\utwoolstar}{\control^{2*,\text{OL}}}

\newcommand{\cost}{\ell}
\newcommand{\lone}{\cost}
\newcommand{\ltwo}{-\cost}
\newcommand{\lonet}{\lone_t}
\newcommand{\ltwot}{\ltwo_t}
\newcommand{\lit}{\cost^i_t}
\newcommand{\lis}{\cost^i_s}
\newcommand{\li}{\cost^i}
\newcommand{\lt}{\cost_t}
\newcommand{\lmit}{\cost^{-i}_t}
\newcommand{\lmis}{\cost^{-i}_s}
\newcommand{\lmi}{\cost^{-i}}

\newcommand{\lamb}{\lambda}
\newcommand{\lambi}{\lambda^i}
\newcommand{\lambmi}{\lambda^{-i}}
\newcommand{\lambit}{\lambi_t}
\newcommand{\lambitop}{\lamb^{i\top}_t}
\newcommand{\lambis}{\lambi_s}
\newcommand{\lambistop}{\lamb^{i\top}_s}

\newcommand{\psimap}{\psi}
\newcommand{\psimapit}{\psimap^i_t}
\newcommand{\psimapitop}{\psimap^{i\top}_t}
\newcommand{\psimapis}{\psimap^i_s}
\newcommand{\psimapistop}{\psimap^{i\top}_s}
\newcommand{\psimapmis}{\psimap^{-i}_s}
\newcommand{\psimapmik}{\psimap^{-i}_K}
\newcommand{\psimapik}{\psimap^{i}_K}

\newcommand{\lagrangianOL}{\mathcal{L}^{OL}}
\newcommand{\lagrangianFB}{\mathcal{L}^{FB}}
\newcommand{\lagrangianiOL}{\mathcal{L}^{i,OL}}
\newcommand{\lagrangianiFB}{\mathcal{L}^{i,FB}}

\newcommand{\map}{\pi}
\newcommand{\mapi}{\map^i}
\newcommand{\mapone}{\map^1}
\newcommand{\maptwo}{\map^2}
\newcommand{\mapmi}{\map^{-i}}
\newcommand{\mapit}{\mapi_t}
\newcommand{\maponet}{\mapone_t}
\newcommand{\maptwot}{\maptwo_t}
\newcommand{\mapmit}{\mapmi_t}
\newcommand{\mapmis}{\mapmi_s}
\newcommand{\mapieq}{\map^{i*}}
\newcommand{\maponeeq}{\map^{1*}}
\newcommand{\maptwoeq}{\map^{2*}}

\newcommand{\val}{V}
\newcommand{\valt}{\val_t}
\newcommand{\valit}{\val^i_t}

\newcommand{\conval}{Z}
\newcommand{\convalt}{\conval_t}
\newcommand{\convalit}{\conval^i_t}

\newcommand{\lb}{\left(}
\newcommand{\rb}{\right)}

\newcommand{\ai}{a^i}
\newcommand{\ami}{a^{-i}}
\newcommand{\ait}{\ai_t}
\newcommand{\amit}{\ami_t}
\newcommand{\ais}{\ai_s}
\newcommand{\amis}{\ami_s}
\newcommand{\bi}{b^i}
\newcommand{\bit}{\bi_t}
\newcommand{\bmi}{b^{-i}}
\newcommand{\bmit}{\bmi_t}
\newcommand{\bis}{\bi_s}
\newcommand{\bmis}{\bmi_s}

\newcommand{\nulow}{\underaccent{\bar}{\nu}}
\newcommand{\nuhigh}{\bar{\nu}}
\newcommand{\lagrangianOLnew}{\tilde{\mathcal{L}}^{OL}}
\newcommand{\lagrangianFBnew}{\tilde{\mathcal{L}}^{FB}}
\newcommand{\lagrangianiOLnew}{\tilde{\mathcal{L}}^{i,OL}}
\newcommand{\lagrangianiFBnew}{\tilde{\mathcal{L}}^{i,FB}}

\newcommand{\lambnew}{\tilde{\lambda}}
\newcommand{\lambinew}{\tilde{\lambda}^i}
\newcommand{\lambminew}{\tilde{\lambda}^{-i}}
\newcommand{\lambitnew}{\lambinew_t}
\newcommand{\lambitopnew}{\lambnew^{i\top}_t}
\newcommand{\lambisnew}{\lambinew_s}
\newcommand{\lambistopnew}{\lambnew^{i\top}_s}

\newcommand{\psimapnew}{\tilde{\psi}}
\newcommand{\psimapitnew}{\psimapnew^i_t}
\newcommand{\psimapitopnew}{\psimapnew^{i\top}_t}
\newcommand{\psimapisnew}{\psimapnew^i_s}
\newcommand{\psimapistopnew}{\psimapnew^{i\top}_s}
\newcommand{\psimapmisnew}{\psimapnew^{-i}_s}
\newcommand{\psimapmiknew}{\psimapnew^{-i}_K}
\newcommand{\psimapiknew}{\psimapnew^{i}_K}

\begin{abstract}
Non-cooperative dynamic game theory provides a principled approach to modeling sequential decision-making among multiple noncommunicative agents. A key focus has been on finding Nash equilibria in two-agent zero-sum dynamic games under various information structures. A well-known result states that in linear-quadratic games, unique Nash equilibria under feedback and open-loop information structures yield identical trajectories. Motivated by two key perspectives---(i) many real-world problems extend beyond linear-quadratic settings and lack unique equilibria, making only \emph{local} Nash equilibria computable, and (ii) local open-loop Nash equilibria (OLNE) are easier to compute than local feedback Nash equilibria (FBNE)---it is natural to ask whether a similar result holds for local equilibria in zero-sum games. To this end, we establish that for a broad class of zero-sum games with potentially nonconvex-nonconcave objectives and nonlinear dynamics: (i) the state/control trajectory of a local FBNE satisfies local OLNE first-order optimality conditions, and vice versa, (ii) a local FBNE trajectory satisfies local OLNE second-order necessary conditions, (iii) a local FBNE trajectory satisfying feedback sufficiency conditions also constitutes a local OLNE, and  (iv) with additional hard constraints on agents' actuations, a local FBNE where strict complementarity holds also satisfies local OLNE first-order optimality conditions, and vice versa.
\end{abstract}

\section{Introduction}\label{sec: intro}
In contrast to single-agent optimization problems, dynamic games require the additional specification of an \emph{information structure}, which is the information available to each agent at every time step of the decision-making process. Information structures have two extremes: open-loop and feedback.
The open-loop information structure assumes that, at each time $t$, every agent knows only the initial state of the game and nothing else about the state. 
On the other hand, the feedback information structure assumes that, at each time $t$, every agent knows the full state of all agents. 
The underlying information structure of a game can greatly affect the existence and expressivity of Nash equilibrium solutions. 
In particular, feedback strategies can encode complex behaviors such as delayed commitment which are not expressible in open-loop strategies \cite{SmoothGameTheory,gupta2023game,Zhan-RSS-21}. 
However, outside of linear-quadratic settings, feedback Nash equilibria are generally far more complicated to compute than open-loop Nash solutions, cf. \cite{laine2023computation}.
\emph{Therefore, it is always valuable to know if for a game, the more computationally-intensive feedback equilibrium differs from an open-loop equilibrium.} 
\par
In the general-sum game setting, it is known that feedback Nash equilibria (FBNE) and open-loop Nash equilibria (OLNE) often diverge greatly \cite{chiu2024extent}. In the two-agent zero-sum setting, while some results comparing FBNE and OLNE do exist, these results either (i) are restricted to linear-quadratic (LQ) games or (ii) assume the existence of a strongly unique Nash equilibrium. \par
However, we observe that these existing results cannot be applied to a large class of zero-sum games that have significant practical applications. In particular, many applications have nonlinear dynamics and more general nonquadratic costs. Such scenarios are often solved by employing iterative algorithms that solve an approximated LQ game at every iteration \cite{starr1969nonzero,fridovich2020efficient,le2022algames,laine2023computation}. In such non-LQ settings, these methods can only find an \emph{approximate local} FBNE/OLNE: it is generally intractable to find a global Nash equilibria, let alone one which is strongly unique. 
\par
Further, existing results that relax the LQ assumption are difficult to verify in potentially nonconvex-nonconcave settings such as generative adversarial network training \cite{goodfellow2014generative}, robust optimization \cite{ben2009robust}, multi-agent reinforcement learning \cite{daskalakis2020independent}, etc. In such settings, a unique Nash equilibrium may not exist for any information structure, or may be intractable to find. In such zero-sum nonconvex-nonconcave settings, and without additional structural assumptions, all existing game-theoretic solvers can only find local Nash equilibria, if they exist \cite{adolphs2019local,mazumdar2019finding,gupta2024second}.

\textbf{Contributions.} With this discussion in mind, it is pertinent to ask: \emph{How are local FBNE and local OLNE related in zero-sum games that extend beyond linear-quadratic settings, and/or do not have a unique equilibrium?} To this end, we present the following contributions, which apply in a large class of zero-sum games with potentially nonconvex-nonconcave costs and non-linear dynamics:
\begin{enumerate}
    \item We show that any local FBNE also satisfies the first-order necessary conditions for a local OLNE of the game, and vice versa.
    \item We show that any local FBNE also satisfy the second-order necessary conditions for a local OLNE of the game. Further, a local FBNE satisfying feedback second-order sufficiency conditions also constitutes a local OLNE of the game.
    \item We show that, in the presence of additional constraints on agents’ control variables, any local FBNE satisfying strict complementarity still satisfies the first-order necessary conditions for a local OLNE of the game, and vice versa.
\end{enumerate}

\section{Related Existing Results}
In the setting of general-sum dynamic games, it is well known that the state trajectories formed by control strategies corresponding to FBNE and OLNE often diverge greatly, and the extent of their divergence has been studied for general-sum linear-quadratic (LQ) dynamic games \cite{chiu2024extent}.
\par
In zero-sum games however, the alignment of these equilibrium concepts becomes more nuanced.  In continuous-time zero-sum pursuit-evasion games \cite{isaacs1999differential}, for example, open-loop strategies are often used to synthesize feedback strategies \cite[Chapter 8]{bacsar1998dynamic}. Moreover, the existing results in the discrete-time zero-sum setting hold only for a limited class of games. For a certain subclass of two-agent zero-sum LQ games, including convex-concave zero-sum LQ games, it is known that a unique FBNE exists and its open-loop realization is also an (not necessarily unique) OLNE \cite[Theorem 6.7]{bacsar1998dynamic}. Further, while the existence of a unique OLNE implies the existence of a unique FBNE in a two-agent zero-sum LQ game, the converse is not true \cite[Proposition 6.2]{bacsar1998dynamic}. Thus, when they both exist, a unique OLNE and a unique FBNE generate the same state trajectory in a two-agent zero-sum LQ game \cite[Remark 6.7]{bacsar1998dynamic}. 
\par
Another existing result applies beyond the setting of LQ games, but in turn requires the game to have a \emph{strongly} unique FBNE. A strongly unique equilibrium is a unique equilibrium where each agent's equilibrium strategy is the unique best response to the other's.
For a zero-sum dynamic game, if a strongly unique FBNE and an OLNE exist, then the OLNE is unique and generates the same state trajectory as the FBNE. 
On the other hand, if a strongly unique OLNE and a FBNE exist for the game, then the two also give the same state trajectory \cite[Theorem 6.9]{bacsar1998dynamic}. 
\par
As mentioned in \Cref{sec: intro},  these existing results \emph{do not cover} practical settings of interest, in which 
(i) the games are not linear-quadratic, or (ii) it is not possible to reason about or compute (strongly) unique equilibria, and only the notions of \emph{local} FBNE and OLNE are computable. 
\section{Preliminaries}\label{sec: Preliminaries}
For $n\in \mathbb{N}$, let $[n]$ denote the set $\{1,\dots,n\}$. 
Consider a two-agent zero-sum dynamic game with a fixed decision-making time horizon of $K$ time steps, with dynamics given by $\xtone=f_t(\xt, \uonet, \utwot),~t\in[K]$, where $\xt \in\bb{R}^n$ denotes the concatenated states of both agents at time $t$, and $\uit\in\bb{R}^{m_i}$ denotes the control action of agent $i$ at time $t$. 
For agent $i$, we denote the other agent as $-i:=\{1,2\} \setminus \{i\}$. For $t\in[K]$, let $\xttoend:=\{\xt,\dots,\x_{K+1}\}$, and $\uttoend^i := \{\uit,\dots,\ui_K\}$. Let the set $\x\left(\uone_{1:K}, \utwo_{1:K}\right)$ denote the state trajectory $\x_{1:K+1}$ obtained by unrolling the set of controls $\uone_{1:K}$ and $\utwo_{1:K}$ according to the dynamics $f_t$, and an initial state $\x_1$. For $t\in[K-1]$, let $T_t = \{t,\dots,K\}$.
\par
Without loss of generality, we assume that at time $t$, agent $1$ minimizes a stage-wise cost $\lonet\left(\xt,\uonet,\utwot\right)\in\bb{R}$ (and thus, agent $2$ minimizes $-\lonet$). Further, let the terminal state cost for agent $1$ be represented by $\cost_{K+1}\left(\xkone\right)$. 
Notably, we assume $\cost_t$ can be nonconvex (and nonquadratic), and that $f_t$ can be nonlinear. The only assumption for $\lit$ is that it is $K$ times differentiable, which is an implicit requirement for finding FBNE \cite{laine2023computation, bacsar1998dynamic}.
We denote the cumulative cost for agent $1$ as $J\left(\uone_{1:K}, \utwo_{1:K}\right) =\sum_{t=1}^K \lonet\lb\xt,\uonet,\utwot\rb + \cost_{K+1}\lb\xkone\rb$
where $\xt\in\x\left(\uone_{1:K}, \utwo_{1:K}\right)$. 
For brevity, we denote the stagewise cost for agent $i$ as $\lit$ and the dynamics as $f_t$, and $\li(\xt, \uit, \umit), f_t(\xt, \uit, \umit)$ denote the cost/dynamics functions with the appropriate order of control arguments.
Let the initial state $\x_1$ be known to both agents.
\subsection{Open-Loop Zero-Sum Games}
A local OLNE of the two-agent zero-sum game is a set of controls $\uiol_{1:K} = \{\uioltime{1},\dots,\uioltime{K}\}, ~i=1,2$, and states $\x(\uoneol_{1:K}, \utwool_{1:K})$ such that 
\begin{align}
    \nonumber J(\uoneol_{1:K}, & \utwo_{1:K}) \leq J(\uoneol_{1:K}, \utwool_{1:K}) \leq J(\uone_{1:K}, \utwool_{1:K}) \\
    &\forall ~\uone_{1:K}\in\mathcal{N}(\uoneol_{1:K}),~ \utwo_{1:K}\in\mathcal{N}(\utwool_{1:K}), \label{eq: OLNE def}
\end{align}
where $\mathcal{N}(\cdot)$ denotes a feasible neighborhood around its argument. A local OLNE is said to be \emph{strict} if the inequalities in \eqref{eq: OLNE def} are strict. Note that strictness is different from and does not imply the uniqueness of a local Nash equilibrium. Finding a local OLNE of a two-agent zero-sum game with only dynamics constraints amounts to locally solving the following equilibrium problem:
\begin{align}
    \nonumber \underbrace{\min_{\x_{2:K+1},~\control_{1:K}^1} J\left(\uone_{1:K}, \utwo_{1:K}\right)}_{\text{Agent}~1}&\quad \underbrace{\min_{\x_{2:K+1},~\control_{1:K}^2}-J\left(\uone_{1:K}, \utwo_{1:K}\right)}_{\text{Agent}~2}, \\
    \text{s.t.}~\xtone=~&f_t(\xt,\uonet,\utwot), ~t\in[K]. \label{eq: ol game}
\end{align}
\subsubsection{Necessary conditions for local OLNE}
 The Karush-Kuhn-Tucker (KKT) conditions \cite{nocedal1999numerical} for \eqref{eq: ol game} give the first-order necessary conditions that the corresponding local OLNE must satisfy, under an appropriate constraint qualification. Introducing Lagrange multipliers $\lambit, t\in[K]$ for the constraints in \eqref{eq: ol game}, we define the Lagrangian for agent $i$ as
 \begin{align}
     \lagrangianiOL := \sum_{t=1}^{K+1} \lit - \sum_{t=1}^{K} \lambitop\left(\xtone - f_t\right),
 \end{align}
which in turn yields the following KKT conditions:
 \begin{align}
     \nabla_{\xt}\lit + \nabla_{\xt}f_t^\top \lambit -\lamb^i_{t-1} =& ~0, ~\forall ~t\in T_2,\label{eq: ol kkt nabla xt}\\
     \nabla_{\uit}\lit + \nabla_{\uit}f_t^\top \lambit =& ~0, ~\forall ~t\in[K], \label{eq: ol kkt nabla ut}\\
     \nabla_{\x_{K+1}}\li_{K+1} - \lamb^i_{K} =&~ 0,\label{eq: ol kkt nabla xfinal}\\
     \xtone - f_t(\xt,\uonet,\utwot) =&~ 0, ~\forall ~t\in[K]. \label{eq: ol kkt dynamics cons}
 \end{align}

\subsubsection{Sufficiency conditions for local OLNE} Nonlinear programming theory gives second-order sufficiency conditions for \eqref{eq: ol game}. If a set of states and controls satisfy these conditions in addition to the necessary KKT conditions \eqref{eq: ol kkt nabla xt}-\eqref{eq: ol kkt dynamics cons}, then they correspond to a local OLNE. Consider the stage-wise Lagrangian for agent $i$
\begin{align}
    \lagrangianiOL_t := \lit - \lambitop\lb\xtone-f_t\rb.
\end{align}
Then the second-order sufficiency conditions are \cite[Theorem 12.6]{nocedal1999numerical}
\begin{align}
   \nonumber &d_{\ui_1}^\top \nabla^2_{\ui_1} d_{\ui_1} +  \sum_{t=2}^K\begin{bmatrix}
        d_{\xt}\\d_{\uit}
    \end{bmatrix}^\top \nabla_{\xt, \uit}^2\lagrangianiOL_t\begin{bmatrix}
        d_{\xt}\\d_{\uit}
    \end{bmatrix}\quad\quad\quad\quad
     \\
     &\quad\quad\quad\quad+ ~d_{\x_{K+1}}^\top \nabla^2_{\xkone} \li_{K+1} d_{\x_{K+1}} > 0\label{eq: ol second order start}\\
    \nonumber&\forall~\{d_{\ui_1}, d_{\xt}, d_{\uit}, d_{\xkone} ~t\in T_2\}~\text{s.t.}\\
    &\quad\quad\quad\,~d_{\xtone} - \nabla_{\xt}f_t d_{\xt}- \nabla_{\uit} f_t d_{\uit} = 0.\label{eq: ol second order end}
\end{align}
Here, $\begin{bmatrix}
        d_{\ui_1}^\top & d_{\xt}^\top & d_{\uit}^\top & d_{\xkone}^\top
    \end{bmatrix}^\top \neq 0,~t\in[K]$ represents a direction in the critical cone of the open-loop equilibrium problem \eqref{eq: ol game}.
\subsection{Feedback Zero-Sum Games}
A FBNE is defined in terms of mappings $\mapit:\bb{R}^n\to\bb{R}^{m_i},~t\in[K], i=1,2$, value functions $\valt:\bb{R}^n\to\bb{R},~t\in[K+1]$, and control-value functions $\convalt:\bb{R}^n\times\bb{R}^{m_1}\times\bb{R}^{m_2}\to\bb{R},~t\in[K+1]$. These are related to each other recursively as
\begin{align}
    \val(\x_{K+1}) &:= \cost_{K+1}\lb\x_{K+1}\rb,\\
    \convalt(\xt,\uonet,\utwot) &:= \cost_t\lb\xt,\uonet,\utwot\rb + \val_{t+1}\left(f_t\lb\xt,\uonet,\utwot\rb\right),\\
    \valt(\xt) &:= \convalt\left(\xt,\maponet\lb\xt\rb,\maptwot\lb\xt\rb\right),~\forall~t\in[K].
\end{align}
While $V_t, Z_t$ correspond to agent $1$, $-V_t, -Z_t$ correspond to agent $2$.
Let the controls and states at a local zero-sum FBNE be denoted by $\uifb_{1:K} = \{\uifbtime{1},\dots,\uifbtime{K}\}, ~i=1,2$, and $\x\left(\uonefb_{1:K}, \utwofb_{1:K}\right)$ respectively. Then the maps $\mapit(\xt)$ are defined to return $\uifbtime{t}$ such that
\begin{align}
    \nonumber \convalt&(\xt, \uonefbtime{t}, \utwot) \leq \convalt(\xt, \uonefbtime{t}, \utwofbtime{t}) \leq \convalt(\xt, \uonet, \utwofbtime{t})\\
    &\forall ~\uonet\in\mathcal{N}(\uonefbtime{t}), ~\utwot\in\mathcal{N}(\utwofbtime{t}), ~t\in[K]. \label{eq: FBNE def}
\end{align}
If the inequalities in \eqref{eq: FBNE def} are strict, the local FBNE is said to be strict. Finding a local FBNE of a two-agent zero-sum game with only dynamics constraints amounts to locally solving the following optimization problem at each time $t\in [K]$ for agent $i$ \cite[Theorem 2.2]{laine2023computation}:
\begin{align}
    \nonumber \min_{\ui_{t:K},~\umifb_{t+1:K},~\x_{t+1:K+1}} ~\sum_{s=t}^K \cost_s\lb\x_s,\ui_s,\umifb_s\rb& + \cost_{K+1}\lb\xkone\rb, \\
    \text{s.t.}~\xsone - f_s\lb\xs, \ui_s, \umifb_s\rb =0,& ~s\in T_t,\label{eq: fb game}\\
    \nonumber \umifb_s - \mapmit\lb\xs\rb =0,& ~s\in T_{t+1}. 
\end{align}
\subsubsection{Necessary conditions for local FBNE} As before, we list the KKT conditions that states and controls corresponding to a local FBNE of \eqref{eq: fb game} should satisfy. Introducing Lagrange multipliers $\lambis, ~s\in T_t$ and $\psimapis,~s\in T_{t+1}$ for the constraints in \eqref{eq: fb game}, we define the Lagrangian for agent $i$ as
\begin{align}
    \nonumber\lagrangianiFB = \sum_{s=t}^{K+1} \lis &- \sum_{s=t}^{K} \lambistop\left(\xsone - f_s\right) \\
    &-\sum_{s=t+1}^K\psimapistop\lb\umifb_s - \mapmit\rb,
\end{align}
which in turn yields the following KKT conditions:
\begin{align}
    \nabla_{\umifbtime{s}}\lis + (\nabla_{\umifbtime{s}}f_s)^\top \lambis- \psimapis =& ~0,~s\in T_{t+1}\label{eq: fb kkt nabla umis},\\
    \nabla_{\uis}\lis + (\nabla_{\uis}f_s)^\top \lambis =& ~0,~s\in T_t,\label{eq: fb kkt nabla uis}\\
    \nabla_{\x_{K+1}}\li_{K+1} -\lamb^i_K =& ~0,\label{eq: fb kkt nabla xfinal}\\
    \x_{s+1} - f_s(\xs,\uis,\umifbtime{s}) =&~ 0, ~\forall ~s\in T_t,\label{eq: fb kkt dyn cons}\\
    \umifb_s - \mapmit(\xs) =&~0, ~\forall~s\in T_{t+1},\label{eq: fb kkt -i feedback pol cons}\\
    \nonumber\nabla_{\xs} \lis - \lamb^i_{s-1} + \lb\nabla_{\xs}f_s\rb^\top\lambis +\lb\nabla_{\xs}\mapmis\rb&^\top\psimapis\\
    =&~ 0,~\forall~s\in T_{t+1}. \label{eq: fb kkt nabla xs}
\end{align}
\subsubsection{Sufficiency conditions for local FBNE} The stage-wise Lagrangian for agent $i$ in the feedback case becomes
\begin{align}
    \nonumber\lagrangianiFB_s = \lis -& \lambistop\lb\xsone-f_s\rb\\
    -&\psimap^{i\top}_s\lb\umifbtime{s}-\mapmis\rb, ~s\in T_{t+1}. \label{eq: fb stagewise lagrangian}
\end{align}
in addition to the KKT conditions \eqref{eq: fb kkt nabla umis}-\eqref{eq: fb kkt nabla xs}, using \eqref{eq: fb stagewise lagrangian}, the second order sufficiency condition for a local FBNE becomes
\begin{align}
    \nonumber d_{\uit}^\top \nabla_{\uit}^2  \lit d_{\uit} + \sum_{s=t+1}^K\begin{bmatrix}
        d_{\xs}\\d_{\uis}\\d_{\umis}
    \end{bmatrix}^\top &\nabla_{\xs, \uis, \umis}^2\lagrangianiFB_s\begin{bmatrix}
        d_{\xs}\\d_{\uis}\\d_{\umis}
    \end{bmatrix}\\
     + ~d_{\x_{K+1}}^\top \nabla^2_{\xkone}& \li_{K+1} d_{\x_{K+1}} > 0\label{eq: fb second order start}\\
    \nonumber\forall~ \{d_{\uit}, d_{\xs}, d_{\uis}, d_{\umis}, d_{\xkone}, &~ s\in T_{t+1}\}~\text{s.t.} \\
     d_{\xsone} - \nabla_{\xs}f_s d_{\xsone}- \nabla_{\uis, \umis}f_s& \begin{bmatrix}
        d_{\uis} \\ d_{\umis}
    \end{bmatrix} = 0~\forall~ s\in T_{t+1}, \label{eq: fb second order xs}
    \\
    d_{\xtone} - \nabla_{\uit}f_s &d_{\uit} = 0,\label{eq: fb second order xt}\\
    d_{\umis} - \nabla_{\xs}\mapmis &d_{\xs} = 0~\forall~ s\in T_{t+1}.\label{eq: fb second order end}
\end{align}
Here, $\begin{bmatrix}
        d_{\uit}^\top & d_{\xs}^\top & d_{\uis}^\top & d_{\umis}^\top & d_{\xkone}^\top
    \end{bmatrix}^\top \neq 0, ~s\in T_{t+1}$ represents a direction in the critical cone of the feedback equilibrium problem \eqref{eq: fb game}.
\begin{remark}\label{remark: fb ol diff}
    For agent $i$, the feedback game \eqref{eq: fb game} differs from the open-loop game \eqref{eq: ol game} by having (i) an equilibrium problem at every time step, and (ii) an additional constraint forcing agent $-i$ to play a feedback policy. Consequently, the feedback KKT conditions \eqref{eq: fb kkt nabla umis}-\eqref{eq: fb kkt nabla xs} and second-order conditions \eqref{eq: fb second order start}-\eqref{eq: fb second order end} differ from the corresponding open-loop KKT and second-order conditions \eqref{eq: ol kkt nabla xt}-\eqref{eq: ol kkt dynamics cons}, \eqref{eq: ol second order start}-\eqref{eq: ol second order end}. However, we observe that if the feedback policy constraints in the feedback game are \emph{weakly active} for both agents at a local FBNE, i.e., $\psimapis=\psimapmis=0$, the feedback KKT conditions reduce to the open-loop KKT conditions. Further, in this case, the feedback constraint in \eqref{eq: fb second order end}, corresponding to the weakly active constraint, can be dropped \cite[Chapter 12]{nocedal1999numerical}. Consequently, in this case, the set of all open-loop critical cone directions becomes a \emph{subset} of the set of all feedback critical cone directions (we establish this in the proof of \Cref{theorem: main theorem}). A similar discussion holds for the second-order necessary conditions for local OLNE and local FBNE, which are similar to the sufficiency conditions with the exceptions of a non-strict inequality and the requirement of an appropriate constraint qualification \cite[Theorem 12.6]{nocedal1999numerical}.
\end{remark}
\section{Main Results}\label{sec: main results}
In the following theorem, we show how state/control trajectories corresponding to any local FBNE policies are related to a local OLNE for the zero-sum game of the type described in \Cref{sec: Preliminaries}, i.e., containing no constraints apart from those due to the game dynamics. 
\begin{theorem}\label{theorem: main theorem}
    Consider a two-agent zero-sum dynamic game with $K$ stages and dynamics constraints. Then:
    \begin{enumerate}
        \item The controls and states corresponding to any existing local FBNE also satisfy the first- and second-order necessary conditions for a local OLNE of the game.
        \item The controls and states corresponding to any existing local OLNE also satisfy the first-order necessary conditions for a local FBNE of the game.
        \item Assume that a local FBNE of the game exists, and that the corresponding local feedback equilibrium strategies $\uifbtstar=\{\mapieq_t(x_t),~t\in [K]\},~i=1,2,$ and state trajectory  $\x(\uonefbstar_{1:K}, \utwofbstar_{1:K})=\{\x^*_1,\dots,\x^*_{K+1}\}$ satisfy the FBNE second-order sufficiency conditions. Then $\uifbtstar$ and $\x(\uonefbstar_{1:K}, \utwofbstar_{1:K})$ also constitute a local OLNE of the game.
    \end{enumerate}
\end{theorem}
\begin{proof}
    We prove the first part by showing that at any local FBNE of the game, the feedback policy constraints \eqref{eq: fb kkt -i feedback pol cons} are weakly active for both agents. The KKT condition for agent $-i$ corresponding to \eqref{eq: fb kkt nabla xfinal} is
    \begin{equation}
    \nabla_{\x_{K+1}}\lmi_{K+1}(\xstartime{K+1}) -\lambmi_K =0. \label{eq: proof mi final state kkt}
    \end{equation}
    Because the game is zero-sum, $\nabla_{\x_{K+1}}\lmi_{K+1} = - \nabla_{\x_{K+1}}\li_{K+1}$, and from \eqref{eq: proof mi final state kkt} combined with \eqref{eq: fb kkt nabla xfinal} evaluated at $\xstartime{K+1}$ we get
    \begin{equation}
        \lambmi_K = -\lambi_K. \label{eq: neg lambdas}
    \end{equation}
    Using the KKT condition \eqref{eq: fb kkt nabla umis} corresponding to agent $-i$ evaluated at time $s=K$, and \eqref{eq: neg lambdas}, we get 
    \begin{align}
        -\nabla_{\uifbtime{K}}\li_K - (\nabla_{\uifbtime{K}}f_K)^\top \lambi_{K}- \psimapmik = ~0. \label{eq: mi umis kkt feedback}
    \end{align}
    Using \eqref{eq: fb kkt nabla uis} evaluated at time $s=K$ and \eqref{eq: mi umis kkt feedback} we get
    \begin{align}
        \psimapmik = 0. \label{eq: psimapmik = 0}
    \end{align}
    Similarly, for the Lagrange multiplier for agent $i$ in the KKT conditions of agent $-i$, we get
    \begin{align}
        \psimapik = 0. \label{eq: psimapik = 0}
    \end{align}
    Using the KKT condition \eqref{eq: fb kkt nabla xs} for both agents $i, -i$ evaluated at time $s=K$, \eqref{eq: neg lambdas}, \eqref{eq: psimapmik = 0} and \eqref{eq: psimapik = 0}, we get
    \begin{align}
        \lambmi_{K-1} = -\lambi_{K-1}. \label{eq: neg lambdas cascading}
    \end{align}
    Equation \eqref{eq: neg lambdas cascading} suggests the following recursive pattern, which can be verified by continuing the previous arguments: $\lambmi_t = -\lambi_t,~\forall~t\in[K]$. This yields
    \begin{align}
        \psimap^i_t = \psimap^{-i}_t = 0~\forall~t\in T_2.\label{eq: all psi zero}
    \end{align}
    \emph{Thus, at a local FBNE, the constraints forcing both agents to play feedback policies are weakly active.} From \Cref{remark: fb ol diff} and \eqref{eq: all psi zero}, we can conclude that the local FBNE KKT conditions \eqref{eq: fb kkt nabla umis}-\eqref{eq: fb kkt nabla xs} imply the local OLNE KKT conditions \eqref{eq: ol kkt nabla xt}-\eqref{eq: ol kkt dynamics cons} and thus $\uonefbstar,~\utwofbstar,~\x(\uonefbstar_{1:K}, \utwofbstar_{1:K})$ also satisfy the first-order necessary conditions for a local OLNE. 
    \par
    Now consider the feedback second-order sufficiency conditions \eqref{eq: fb second order start}-\eqref{eq: fb second order end}. Because the feedback policy constraints are weakly active, from \Cref{remark: fb ol diff}, the condition \eqref{eq: fb second order end} can be dropped while considering the directions in the critical cone for which \eqref{eq: fb second order start} should hold. Further, we observe that the directions such that
    \begin{align}
    d_{\umit} & =0  \label{eq: fb reduced second order umit}\\
    d_{\xtone} - \nabla_{\uit}f_s d_{\uit} &= 0, \label{eq: fb reduced second order xt}\\
    d_{\xsone} - \nabla_{\xs}f_s d_{\xsone}- \nabla_{\uis}f_s 
        d_{\uis}
     &= 0~\forall~ s\in T_{t+1}, \label{eq: fb reduced second order xs}
    \end{align}
    all lie in the directions in the critical cone defined by \eqref{eq: fb second order xs}, \eqref{eq: fb second order xt}. We observe that \eqref{eq: fb reduced second order umit}-\eqref{eq: fb reduced second order xt} when applied to \eqref{eq: fb second order start} exactly yield the open-loop second-order sufficiency conditions \eqref{eq: ol second order start}-\eqref{eq: ol second order end}.  \emph{Thus the directions in the open-loop critical cone are a subset of the directions in the feedback critical cone.} From \Cref{remark: fb ol diff} a similar relation holds for second order necessary conditions between feedback and open-loop settings. Because the FBNE controls/states must be satisfying the feedback second-order necessary conditions, this implies that they also satisfy the OLNE second order necessary conditions. This proves the first part of the theorem.
    \par
    To prove the second part of the theorem, let a local OLNE exist and $\uoneolstar_{1:K}, \utwoolstar_{1:K}, \x^{*,OL}_{1:K+1}:=\x(\uoneolstar_{1:K}, \utwoolstar_{1:K})$ denote the corresponding controls and states. Thus, $\uoneolstar_{1:K}, \utwoolstar_{1:K}, \x^{*,OL}$ satisfy the open-loop KKT conditions \eqref{eq: ol kkt nabla xt}-\eqref{eq: ol kkt dynamics cons}. Then it is trivial to observe that $\uoneolstar_{1:K}, \utwoolstar_{1:K}, \x^{*,OL}$ also satisfy the feedback KKT conditions \eqref{eq: fb kkt nabla umis}-\eqref{eq: fb kkt nabla xs} for the choice of feedback policy constraints $\psimap^i_t = \psimap^{-i}_t = 0~\forall~t\in[K]$. Note that the open-loop trajectory also satisfies \eqref{eq: fb kkt -i feedback pol cons}, because the open-loop controls are derived by solving the KKT conditions backwards in time, which implies that the local OLNE controls satisfy \eqref{eq: FBNE def}.
    \par
    For the third part of the theorem, if $\uonefbstar,~\utwofbstar,~\x(\uonefbstar_{1:K}, \utwofbstar_{1:K})$ satisfy the feedback second-order sufficiency conditions \eqref{eq: fb second order start}-\eqref{eq: fb second order end}, from the discussion above, we can conclude that $\uonefbstar,~\utwofbstar,~\x(\uonefbstar_{1:K}, \utwofbstar_{1:K})$ also satisfy the second-order sufficiency conditions of (and thus constitute) a local OLNE of the game.
\end{proof}
\section{A Broader Class of Zero-Sum Dynamic Games}
The results in \Cref{sec: main results} rely on a cascading pattern which ensures that the feedback policy constraints \eqref{eq: fb kkt -i feedback pol cons} are only weakly active at a local FBNE. While true for a zero-sum game with only dynamics constraints, this pattern may not hold true for zero-sum games with more arbitrary state/control constraints in general. In this section, however, we identify a broad class of games in which a similar property still arises. 

Consider the class of two-agent zero-sum games with (i) dynamics constraints, and (ii) control bound (inequality) constraints for each agent. Such bound constraints often occur in practical applications, for example, in systems with hard actuator limits.

For such a game, the corresponding open-loop equilibrium problem for agent $i$ is
\begin{align}
    \nonumber &\min_{\x_{2:K+1},~\ui_{1:K}} \sum_{t=1}^K \lit(\xt,\uit,\umit) + \li_{K+1}(\xkone)\\
    \nonumber\text{s.t.}~\,&\xtone=~f_t(\xt,\uonet,\utwot), ~t\in[K],\\
    &\quad\quad\quad\,\,\,\ait\leq\uit\leq\bit, ~t\in[K].  \label{eq: new ol game}
\end{align}
Using Lagrange multipliers $\lambitnew, \nulow^i_t, \nuhigh^i_t,~t \in [K]$, the Lagrangian for agent $i$ defined as
\begin{align}
    \nonumber\lagrangianiOLnew := \sum_{t=1}^{K+1} \lit &- \sum_{t=1}^{K} \bigg(\lambitopnew\left(\xtone - f_t\right)\\
    &+ \nulow^{i\top}_t\lb\uit-\ait\rb + \nuhigh^{i\top}_t\lb\bit-\uit\rb\bigg).
\end{align}
The feedback equilibrium problem for agent $i$ becomes
\begin{align}
    \nonumber \min_{\ui_{t:K},~\umifb_{t+1:K},~\x_{t+1:K+1}} ~\sum_{s=t}^K \cost_s\lb\x_s,\ui_s,\umifb_s\rb& + \cost_{K+1}\lb\xkone\rb, \\
    \nonumber\text{s.t.}~\xsone - f_s\lb\xs, \ui_s, \umifb_s\rb =0,& ~s\in T_t,\\
    \nonumber \umifb_s - \mapmit\lb\xs\rb =0,& ~s\in T_{t+1},\\
    \ais\leq\uifb_s\leq\bis,& ~s\in T_{t}.\label{eq: new fb game}
\end{align}
Note that we do not encode agent $-i$'s control constraints in agent $i$'s problem at time $t$, as these constraints will be implicitly imposed through agent $-i$'s problem at future time steps. Using Lagrange multipliers $\lambisnew, \nulow^i_s, \nuhigh^i_s,~s \in T_t$ and $\psimapisnew, s\in T_{t+1}$, the Lagrangian corresponding to \eqref{eq: new fb game} is defined as
\begin{align}
    &\lagrangianiFBnew = \sum_{s=t}^{K+1} \lis -\sum_{s=t+1}^K\psimapistopnew\lb\umifb_s - \mapmit\rb  \\
    -\sum_{s=t}^{K}&\lb \lambistopnew\left(\xsone - f_s\right) + \nulow^{i\top}_s\lb\uis - \ais\rb + \nuhigh^{i\top}_s\lb\bis - \uis\rb \rb.\nonumber
\end{align}

In such constrained games with nonconvex-nonconcave objectives and nonlinear dynamics, it is usually only possible to find solutions that (approximately) satisfy the FBNE first-order necessary conditions using existing game-theoretic feedback methods \cite{laine2023computation, fridovich2020efficient}. These first-order feedback solutions also require much more computation than that required to find solutions satisfying OLNE first-order necessary conditions. Existing open-loop game-theoretic solvers can find exact solutions that satisfy OLNE first-order (and second-order) necessary conditions in such games \cite{gupta2024second, mazumdar2019finding, adolphs2019local}. Thus, it is valuable to know if some relation between local OLNE and local FBNE akin to \Cref{theorem: main theorem} holds in such games as well.
\par 
To this end, we show that when strict complementarity holds (a standard assumption) in such zero-sum games, the local OLNE first-order necessary conditions are the same as the local FBNE first-order necessary conditions.

\begin{theorem}
    Consider a two-agent zero-sum dynamic game with $K$ stages, dynamic constraints, and control bound (inequality) constraints for each agent. Then:
    \begin{enumerate}
        \item The controls and states corresponding to any existing local FBNE where strict complementarity holds also satisfy the first-order necessary conditions for a local OLNE of the game.
        \item The controls and states corresponding to any existing local OLNE where strict complementarity holds also satisfy the first-order necessary conditions for a local FBNE of the game.
    \end{enumerate}
\end{theorem}
\begin{proof}
    The local OLNE first-order necessary conditions for \eqref{eq: new ol game} are
     \begin{align}
     \nabla_{\xt}\lit + \nabla_{\xt}f_t^\top \lambitnew -\lambnew^i_{t-1} =& ~0, ~\forall ~t\in T_2,\label{eq: new ol kkt nabla xt}\\
     \nabla_{\uit}\lit + \nabla_{\uit}f_t^\top \lambitnew - \nulow^i_t + \nuhigh^i_t =& ~0, ~\forall ~t\in[K], \label{eq: new ol kkt nabla ut}\\
     \nabla_{\x_{K+1}}\li_{K+1} - \lambnew^i_{K} =&~ 0,\label{eq: new ol kkt nabla xfinal}\\
     \xtone - f_t(\xt,\uonet,\utwot) =&~ 0, ~\forall ~t\in[K], \label{eq: new ol kkt dynamics cons}\\
     \ait \leq \uit \perp \nulow^i_t\geq 0,&~\forall ~t\in[K],\label{eq: new ol kkt low u cons}\\
     0 \leq \nuhigh^i_t \perp \uit \leq \bit,&~\forall ~t\in[K].\label{eq: new ol kkt high u cons}
 \end{align}
 Similarly, the local FBNE first-order necessary conditions for \eqref{eq: new fb game} are
 \begin{align}
         \nabla_{\umifbtime{s}}\lis + (\nabla_{\umifbtime{s}}f_s)^\top \lambisnew- \psimapisnew =& ~0,~s\in T_{t+1}\label{eq: new fb kkt nabla umis},\\
    \nabla_{\uis}\lis + (\nabla_{\uis}f_s)^\top \lambisnew - \nulow^i_s + \nuhigh^i_s =& ~0,~s\in T_t,\label{eq: new fb kkt nabla uis}\\
    \nabla_{\x_{K+1}}\li_{K+1} -\lambnew^i_K =& ~0,\label{eq: new fb kkt nabla xfinal}\\
    \x_{s+1} - f_s(\xs,\uis,\umifbtime{s}) =&~ 0, ~\forall ~s\in T_t,\label{eq: new fb kkt dyn cons}\\
    \umifb_s - \mapmit(\xs) =&~0, ~\forall~s\in T_{t+1},\label{eq: new fb kkt -i feedback pol cons}\\
    \ais \leq \uis \perp \nulow^i_s\geq 0&,\,\,\quad\forall ~s\in T_t,\label{eq: new fb kkt low u cons}\\
     0\leq\nuhigh^i_s \perp \uis \leq \bis&,\quad\,\,\forall ~s\in T_t,\label{eq: new fb kkt high u cons}\\
    \nonumber\nabla_{\xs} \lis - \lambnew^i_{s-1} + \lb\nabla_{\xs}f_s\rb^\top\lambisnew +\lb\nabla_{\xs}\mapmis\rb&^\top\psimapisnew\\
    =&~ 0,~\forall~s\in T_{t+1}. \label{eq: new fb kkt nabla xs}
 \end{align}
 If none of the control bound constraints are active at a local FBNE/OLNE for both agents, then the result trivially follows from \Cref{theorem: main theorem}. Thus, we need to analyze the case when at least one control bound constraint is active for an agent. 
 \par
 To prove the first part, consider the feedback necessary conditions \eqref{eq: new fb kkt nabla umis}-\eqref{eq: new fb kkt nabla xs}. From \eqref{eq: new fb kkt nabla xfinal} and the corresponding condition for agent $-i$, we get
 \begin{align}
     \lambnew^i_K = -\lambnew^{-i}_K. \label{eq: new neg lambdas}
 \end{align}
 Through \eqref{eq: new fb kkt nabla umis}, \eqref{eq: new neg lambdas} evaluated at $s=K$, and the condition for agent $-i$ corresponding to \eqref{eq: new fb kkt nabla uis} evaluated at $s=K$, we get
 \begin{align}
     \psimapiknew = \nuhigh^{-i}_K - \nulow^{-i}_K. \label{eq: new psi i}
 \end{align}
 Similarly, for agent $-i$, we get
 \begin{align}
     \psimapmiknew = \nuhigh^{i}_K - \nulow^{i}_K. \label{eq: new psi mi}
 \end{align}
 Let $[q]_j$ denote the $j^\text{th}$ component of some vector $q$, and let $[Q]_j$ denote the $j^\text{th}$ row of some matrix $Q$. Consider $[\umi_K]_j$ and $[\nabla_{\x_K}\mapmi_K]_j, ~j\in[m^{-i}]$. Due to strict complementarity, either one of the bound constraints on $[\umi_K]_j$ is active and $[\nabla_{\x_K}\mapmi_K]_j=0$, or both bounds on $[\umi_K]_j$ are inactive and $[\nulow^i_K]_j=[\nuhigh^i_K]_j = 0$. A similar argument holds for agent $-i$'s problem. Using \eqref{eq: new psi i}, \eqref{eq: new psi mi}, this implies that at any local FBNE of the game,
 \begin{align}
     \lb\nabla_{\x_K}\mapmi_K\rb^\top\psimapiknew =  \lb\nabla_{\x_K}\mapi_K\rb^\top\psimapmiknew = 0.  \label{eq: new fbne phenomenon}
 \end{align}
 Substituting \eqref{eq: new fbne phenomenon} into \eqref{eq: new fb kkt nabla xs} and using \eqref{eq: new neg lambdas}, we get $\lambnew^i_{K-1}=-\lambnew^{-i}_{K-1}$. This leads to a cascading pattern backwards in time as before, leading to
 \begin{align}
     \lb\nabla_{\x_t}\mapmi_K\rb^\top\psimapiknew =  \lb\nabla_{\x_t}\mapi_K\rb^\top\psimapmiknew = 0~ \forall~t\in T_2.\label{eq: new fbne phenomenon all times}
 \end{align}
 \Cref{eq: new fbne phenomenon all times} implies that at a local FBNE, the OLNE first-order necessary conditions \eqref{eq: new ol kkt nabla xt}-\eqref{eq: new ol kkt high u cons} become a subset of the FBNE first-order necessary conditions \eqref{eq: new fb kkt nabla umis}-\eqref{eq: new fb kkt nabla xs}. Thus, the controls/states at any local FBNE also satisfy the first-order necessary conditions for a OLNE of the game. 
 \par
 To prove the second part, let a local OLNE exist and $\uoneolstar_{1:K}, \utwoolstar_{1:K}, \x^{*,OL}_{1:K+1}:=\x(\uoneolstar_{1:K}, \utwoolstar_{1:K})$ denote the corresponding controls and states. Thus, $\uoneolstar_{1:K}, \utwoolstar_{1:K}$ and $\x^{*,OL}_{1:K+1}$ satisfy the open-loop conditions \eqref{eq: new ol kkt nabla xt}-\eqref{eq: new ol kkt high u cons}. Then through the discussion above, it is trivial to observe that $\uoneolstar_{1:K}, \utwoolstar_{1:K}$ and $\x^{*,OL}_{1:K+1}$ also satisfy the feedback first-order necessary conditions \eqref{eq: new fb kkt nabla umis}-\eqref{eq: new fb kkt nabla xs} for a choice of $\psimapnew^i_t = \nuhigh^{-i}_t - \nulow^{-i}_t, \psimapnew^{-i}_t = \nuhigh^{i}_t - \nulow^{i}_t, t\in T_2$. Thus, $\uoneolstar_{1:K}, \utwoolstar_{1:K}$ and $\x^{*,OL}_{1:K+1}$ satisfy the first-order necessary conditions for a local FBNE of the game.
\end{proof}
\section{Conclusion}
In dynamic games and outside of highly-structured (e.g., linear-quadratic) cases, a feedback Nash equilibrium (FBNE) is much more computationally expensive to compute than an open-loop equilibrium (OLNE). Thus, it is always valuable to know if a relation between the two equilibrium types exists. In the context of zero-sum games, existing results state that (i) in the linear-quadratic zero-sum game setting, unique feedback and open-loop equilibria generate the same control/state trajectory, and (ii) in a (potentially not linear-quadratic) zero-sum game, if a strongly unique FBNE and a OLNE exists, then they generate the same control/state trajectory. However, many practical applications of interest have nonquadratic, potentially nonconvex-nonconcave objectives along with nonlinear dynamics. In such settings, the (strong) uniqueness of a Nash equilibrium generally cannot be verified, and only \emph{local} Nash equilibrium are computable. Thus, existing results cannot be applied to a large class of zero-sum games with practical applications. To this end, we show that for this large class of zero-sum games:
\begin{enumerate}
    \item Any local FBNE also satisfies the first-order necessary conditions for a local OLNE of the game, and vice versa.
    \item Any local FBNE also satisfies the second-order necessary conditions for a local OLNE of the game. Further, a local FBNE satisfying feedback sufficiency conditions also constitutes a local OLNE of the game.
    \item In the presence of additional constraints on agents’ control variables, any local FBNE (where strict complementarity holds) still satisfies the first-order necessary conditions for a local OLNE of the game, and vice versa.
\end{enumerate}

\addtolength{\textheight}{-12cm}   





\section*{ACKNOWLEDGMENT}
The authors would like to thank Jesse Milzman for insightful discussions on related topics.



\bibliographystyle{IEEEtran}
\bibliography{IEEEabrv,references}

\end{document}